\documentclass{easychair}
\usepackage{isabelle}
\usepackage{isabellesym}
\usepackage{verbatim}
\usepackage{xcolor}
\colorlet{snippet}{urlcolor}

\newcommand\MnFont[1]{
 \DeclareFontFamily{U}{MnSymbol#1}{}
 \DeclareSymbolFont{MnSy#1}{U}{MnSymbol#1}{m}{n}
 \SetSymbolFont{MnSy#1}{bold}{U}{MnSymbol#1}{b}{n}
 \DeclareFontShape{U}{MnSymbol#1}{m}{n}{
     <-6>  MnSymbol#15
    <6-7>  MnSymbol#16
    <7-8>  MnSymbol#17
    <8-9>  MnSymbol#18
    <9-10> MnSymbol#19
   <10-12> MnSymbol#110
   <12->   MnSymbol#112}{}
 \DeclareFontShape{U}{MnSymbol#1}{b}{n}{
     <-6>  MnSymbol#1-Bold5
    <6-7>  MnSymbol#1-Bold6
    <7-8>  MnSymbol#1-Bold7
    <8-9>  MnSymbol#1-Bold8
    <9-10> MnSymbol#1-Bold9
   <10-12> MnSymbol#1-Bold10
   <12->   MnSymbol#1-Bold12}{}
}
\MnFont{A}\MnFont{B}\MnFont{C}\MnFont{D}
\DeclareMathSymbol{\mncirc}{\mathbin}{MnSyC}{'130}

\def\test#1#2#3{\setbox0=\hbox{$\vphantom{#1}^{#2}_{#3}$}%
                \dimen0=\wd0%
                \setbox1=\hbox{$\scriptstyle #2$}%
                \advance\dimen0-\wd1%
                \setbox1=\hbox{\hskip\dimen0\copy1}%
                \dimen0=\wd0%
                \setbox2=\hbox{$\scriptstyle #3$}%
                \advance\dimen0-\wd2%
                \setbox2=\hbox{\hskip\dimen0\copy2}%
                {\vphantom{#1}^{\box1}_{\box2}}{#1}
}

\newcommand{\overlayrel}[4]{\mathrel{%
 \def\next##1##2{%
  \setbox0=\hbox{$##1#3$}%
  \setbox1=\hbox to\wd0{$##1\hfil\mkern#1mu#4\mkern#2mu\hfil$}%
  \dp1=\dp0\ht1=\ht0\wd0=0pt\box0\box1%
 }%
 \mathpalette\next{}%
}}

\newcommand{\dstep}{\overlayrel02{\to}\mncirc}

\theoremstyle{plain}
\newtheorem{theorem}{Theorem}
\newtheorem{lemma}[theorem]{Lemma}

\theoremstyle{definition}
\newtheorem{definition}[theorem]{Definition}
\newtheorem{example}[theorem]{Example}

\title{A Short Mechanized Proof of the Church-Rosser Theorem
by the Z-property for the $\lambda\beta$-calculus in Nominal Isabelle%
\thanks{This work was partially supported by FWF (Austrian Science Fund) projects P27502 and P27528.}}

\author{
Julian Nagele
\and
Vincent van Oostrom
\and
Christian Sternagel
}

\institute{
  University of Innsbruck,
  Austria\\
  \email{\{julian.nagele,vincent.van-oostrom,christian.sternagel\}@uibk.ac.at}
 }

\authorrunning{J. Nagele, V. van Oostrom, and C. Sternagel}

\titlerunning{A Short Mechanized Proof of CR}

\isabellestyleliteral
\renewcommand\isastyletxt{\isastyletext}
\renewcommand\isastyle{\isastyleminor}

\newcommand\afpurl{https://bitbucket.org/isa-afp/afp-2016/src}
\newcommand\afphgid{8e42416c3a313db6ea1967ce2e44b4e8bee40a0a}
\newcommand\afpref[2]{%
  \href{\afpurl/\afphgid/thys/Rewriting_Z/Lambda_Z.thy?at=default&fileviewer=file-view-default\#Lambda_Z.thy-#1}{#2}}

\newcommand{\DefineSnippet}[2]{%
  \expandafter\newcommand\csname snippet--#1\endcsname{#2}}
\newcommand\Snippet[2][1]{%
  \begin{quote}
  \afpref{#1}{%
  \begin{minipage}{\textwidth}%
  \begin{isabelle}%
  \csname snippet--#2\endcsname
  \end{isabelle}
  \end{minipage}}
  \end{quote}}
\newcommand\snippet[2][1]{{%
  \afpref{#1}{{\isastyle\csname snippet--#2\endcsname}}%
}}
\newcommand\subst[3]{#3\,\ensuremath{[}#1\ensuremath{{}:={}}#2\ensuremath{]}}
\newcommand\abs[2]{\ensuremath{\lambda}#1.\,#2}
\newcommand\app{\ }
\newcommand\appbeta{\:\ensuremath{\cdot_\beta}\:}
\newcommand\bstep[2]{#1 \ensuremath{\to_\beta} #2}
\newcommand\bsteps[2]{#1 \ensuremath{\to_\beta^*} #2}
\newcommand\bstepe[2]{#1 \ensuremath{\to_\beta^=} #2}
\newcommand\fresh{\ \ensuremath{\sharp}\ }

\DefineSnippet{term_type}{%
\isacommand{nominal{\isacharunderscore}datatype}\isamarkupfalse%
\ {\isachardoublequoteopen}term{\isachardoublequoteclose}\ {\isacharequal}\isanewline
\ \ Var\ name\isanewline
{\isacharbar}\ App\ {\isachardoublequoteopen}term{\isachardoublequoteclose}\ {\isachardoublequoteopen}term{\isachardoublequoteclose}\isanewline
{\isacharbar}\ Abs\ x{\isacharcolon}{\isacharcolon}name\ t{\isacharcolon}{\isacharcolon}{\isachardoublequoteopen}term{\isachardoublequoteclose}\ \isakeyword{binds}\ x\ \isakeyword{in}\ t%
}
\DefineSnippet{subst}{%
\begin{isabelle}%
\subst{x}{s}{y}\ {\isacharequal}\ {\isacharparenleft}\isakeyword{if}\ x\ {\isacharequal}\ y\ \isakeyword{then}\ s\ \isakeyword{else}\ y{\isacharparenright}\isasep\isanewline%
\subst{x}{s}{{\isacharparenleft}t\app{}u{\isacharparenright}}\ {\isacharequal}\ \subst{x}{s}{t}\app{}\subst{x}{s}{u}\isasep\isanewline%
y\ {\isasymsharp}\ {\isacharparenleft}x{\isacharcomma}\ s{\isacharparenright}\ {\isasymLongrightarrow}\ \subst{x}{s}{{\isacharparenleft}\abs{y}{t}{\isacharparenright}}\ {\isacharequal}\ \abs{y}{\subst{x}{s}{t}}%
\end{isabelle}%
}
\DefineSnippet{subst_lemma}{%
\begin{isabelle}%
x\ {\isasymsharp}\ {\isacharparenleft}y{\isacharcomma}\ u{\isacharparenright}\ {\isasymLongrightarrow}\ \subst{y}{u}{\subst{x}{s}{t}}\ {\isacharequal}\ \subst{x}{\subst{y}{u}{s}}{\subst{y}{u}{t}}%
\end{isabelle}
}
\DefineSnippet{beta}{%
\begin{isabelle}%
x\ {\isasymsharp}\ t\ {\isasymLongrightarrow}\ \bstep{{\isacharparenleft}\abs{x}{s}{\isacharparenright}\app{}t}{\subst{x}{t}{s}}\isasep\isanewline%
\bstep{s}{t}\ {\isasymLongrightarrow}\ \bstep{s\app{}u}{t\app{}u}\isasep\isanewline%
\bstep{s}{t}\ {\isasymLongrightarrow}\ \bstep{u\app{}s}{u\app{}t}\isasep\isanewline%
\bstep{s}{t}\ {\isasymLongrightarrow}\ \bstep{\abs{x}{s}}{\abs{x}{t}}%
\end{isabelle}
}
\DefineSnippet{beta_rule}{%
\begin{isabelle}%
x\ {\isasymsharp}\ t\ {\isasymLongrightarrow}\ \bstep{{\isacharparenleft}\abs{x}{s}{\isacharparenright}\app{}t}{\subst{x}{t}{s}}%
\end{isabelle}
}
\DefineSnippet{betas_cong}{%
\begin{isabelle}%
\bsteps{s}{t}\ {\isasymLongrightarrow}\ \bsteps{u}{v}\ {\isasymLongrightarrow}\ \bsteps{s\app{}u}{t\app{}v}\isasep\isanewline%
\bsteps{s}{t}\ {\isasymLongrightarrow}\ \bsteps{\abs{x}{s}}{\abs{x}{t}}\isasep\isanewline%
\bsteps{s}{s{\isacharprime}}\ {\isasymLongrightarrow}\ \bsteps{t}{t{\isacharprime}}\ {\isasymLongrightarrow}\ \bsteps{\subst{x}{s}{t}}{\subst{x}{s{\isacharprime}}{t{\isacharprime}}}%
\end{isabelle}
}
\DefineSnippet{lambda_bullet}{%
\begin{isabelle}%
x\isactrlsup {\isasymbullet}\ {\isacharequal}\ x\isasep\isanewline%
{\isacharparenleft}\abs{x}{t}{\isacharparenright}\isactrlsup {\isasymbullet}\ {\isacharequal}\ \abs{x}{t\isactrlsup {\isasymbullet}}\isasep\isanewline%
{\isacharparenleft}s\app{}t{\isacharparenright}\isactrlsup {\isasymbullet}\ {\isacharequal}\ s\isactrlsup {\isasymbullet}\appbeta{}t\isactrlsup {\isasymbullet}%
\end{isabelle}%
}
\DefineSnippet{app_beta}{%
\begin{isabelle}%
x\ {\isasymsharp}\ u\ {\isasymLongrightarrow}\ {\isacharparenleft}\abs{x}{s{\isacharprime}}{\isacharparenright}\appbeta{}u\ {\isacharequal}\ \subst{x}{u}{s{\isacharprime}}\isasep\isanewline%
x\appbeta{}u\ {\isacharequal}\ x\app{}u\isasep\isanewline%
{\isacharparenleft}s\app{}t{\isacharparenright}\appbeta{}u\ {\isacharequal}\ s\app{}t\app{}u%
\end{isabelle}%
}
\DefineSnippet{app_beta_Betas}{%
  \isa{\bsteps{s}{s{\isacharprime}}\ {\isasymLongrightarrow}\ \bsteps{t}{t{\isacharprime}}\ {\isasymLongrightarrow}\ \bsteps{s\appbeta{}t}{s{\isacharprime}\appbeta{}t{\isacharprime}}}%
}
\DefineSnippet{self}{%
  \isa{\bsteps{t}{t\isactrlsup {\isasymbullet}}}%
}
\DefineSnippet{rhs}{%
  \isa{\bsteps{\subst{x}{s\isactrlsup {\isasymbullet}}{t\isactrlsup {\isasymbullet}}}{\subst{x}{s}{t}\isactrlsup {\isasymbullet}}}%
}
\DefineSnippet{lambda_Z}{%
  \isa{\bstep{s}{t}\ {\isasymLongrightarrow}\ \bsteps{t}{s\isactrlsup {\isasymbullet}}\ {\isasymand}\ \bsteps{s\isactrlsup {\isasymbullet}}{t\isactrlsup {\isasymbullet}}}%
}
\DefineSnippet{Abs_BetasE}{%
  \isa{\bsteps{\abs{x}{s}}{t}\ {\isasymLongrightarrow}\ {\isasymexists}u{\isachardot}\ t\ {\isacharequal}\ \abs{x}{u}\ {\isasymand}\ \bsteps{s}{u}}%
}
\DefineSnippet{bullet_App}{%
  \isa{\bstepe{s\isactrlsup {\isasymbullet}\app{}t\isactrlsup {\isasymbullet}}{{\isacharparenleft}s\app{}t{\isacharparenright}\isactrlsup {\isasymbullet}}}%
}

\begin{document}

\maketitle

\begin{abstract}
We present a short proof of the Church-Rosser property for the lambda-calculus
enjoying two distinguishing features: firstly, it employs the Z-property, 
resulting in a short and elegant proof;
and secondly, it is formalized in the nominal higher-order logic
available for the proof assistant Isabelle/HOL.
\end{abstract}

\section{Introduction}

Dehornoy proved confluence for the rule of self-distributivity
$xyz \to xz(yz)$\footnote{%
Confluence of this single-rule term rewrite system is non-trivial: presently 
no tool can prove it automatically.
}
by means of a novel method~\cite{D00}, the idea being to give a map $\bullet$
that is \emph{monotonic} with respect to $\to^\ast$ and that yields for each
object an \emph{upper bound} on all objects reachable from it in a single step.
Later, this method was extracted and dubbed the Z-property~\cite{Deho:Oost:08},
and applied to prove confluence of various rewrite systems, in particular the
$\lambda\beta$-calculus.

Here we present our Isabelle/HOL~\cite{Isabelle} formalization 
of part of the above mentioned work,\footnote{The formalization follows
the pen-and-paper proof \emph{exactly}, except for one mistake in Lemma~\ref{lem:rhs} (Rhs).}
in particular that the $\lambda\beta$-calculus
is confluent since it enjoys the Z-property and that the latter property 
is equivalent to an abstract version of Takahashi's confluence method~\cite{T95}.
We achieve a rigorous treatment of terms modulo $\alpha$-equivalence by employing 
Nominal Isabelle~\cite{UK12}, a nominal higher-order logic based on Isabelle/HOL.
Our formalization is available from the \emph{archive of formal proofs}~\cite{FNOS16}.
Below, Isabelle code-snippets are in blue and hyperlinked.

\section{Nominal $\lambda$-terms}

In our formalization $\lambda$-terms are represented by the following nominal data
type, where the annotation ``\isakeyword{binds} \isa{x} \isakeyword{in} \isa{t}'' indicates that the equality of such abstraction terms is up to renaming of $x$ in $t$:
\Snippet[24]{term_type}
For the sake of readability we will use standard notation, i.e., $x$ instead of
$\textit{Var}~x$, $s\app t$ instead of $\textit{App}~s~t$, and $\abs{x}{t}$
instead of $\textit{Abs}~x~t$, in the remainder.
When defining (recursive) functions on $\lambda$-terms, we may have to take care
of so-called \emph{freshness constraints}. A freshness constraint is written $x
\fresh t$ and states that $x$ does not occur in $t$, or equivalently, $x$ is
fresh for $t$. 
\begin{definition}
Capture-avoiding substitution is defined recursively by the following equations:
\Snippet[83]{subst}
Due to the constraint, the final equation is only applicable when $y$ is fresh for $x$ and $s$.
\end{definition}
In principle it is always possible to rename variables in terms (or any finitely
supported structure) apart from a given finite collection of variables. In
order to relieve the user of doing so by hand, Nominal Isabelle~\cite{UK12}
provides infrastructure for defining nominal functions, giving rise to strong
induction principles that take care of appropriate renaming. (However, nominal
functions do not come for free: after stating the defining equations, we are
faced with proof obligations that ensure pattern-completeness, termination,
equivariance, and well-definedness. With the help of some home-brewed
Eisbach~\cite{Eisbach} methods we were able to handle those obligations
automatically.) We first consider the Substitution Lemma, cf.~\cite[Lemma~2.1.16]{Bare:84}.
\begin{lemma}
\Snippet[100]{subst_lemma}
\end{lemma}
\begin{proof}
In principle the proof proceeds by induction on $t$. However, for the case of
$\lambda$-abstractions we additionally want the bound variable to be fresh for
$s$, $u$, $x$, and $y$. With Nominal Isabelle it is enough to indicate that the
variables of those terms should be avoided in order to obtain appropriately
renamed bound variables. We will not mention this fact again in future
proofs.
\begin{itemize}
\item
In the base case $t = z$ for some variable $z$. If $z = x$ then
$\subst{y}{u}{\subst{x}{s}{t}} = \subst{y}{u}{s}$
and
$\subst{x}{\subst{y}{u}{s}}{\subst{y}{u}{t}} = \subst{y}{u}{s}$,
since then $z
\neq y$ and thus $\subst{y}{u}{z} = z$. Otherwise $z \neq x$. Now if $z = y$,
then
$\subst{y}{u}{\subst{x}{s}{t}} = u$
and
$\subst{x}{\subst{y}{u}{s}}{\subst{y}{u}{t}} = u$, since $x \fresh u$.
If $z \neq y$ then both ends of the equation reduce to $z$ and we are done.

\item
In case of an application, we conclude by definition and twice the IH.

\item
Now for the interesting case. Let $t = \abs{z}{v}$ such that
$z \fresh (s, u, x, y)$. Then
\begin{align*}
\subst{y}{u}{\subst{x}{s}{(\abs{z}{v})}}
&= \abs{z}{\subst{y}{u}{\subst{x}{s}{v}}}
  &&\text{since $z \fresh (s,u,x,y)$}\\
&= \abs{z}{\subst{x}{\subst{y}{u}{s}}{\subst{y}{u}{v}}}
  &&\text{by IH}\\
&= \subst{x}{\subst{y}{u}{s}}{\subst{y}{u}{(\abs{z}{v})}}
  &&\text{since $z \fresh (\subst{y}{u}{s},u,x,y)$}
\end{align*}
where in the last step $z \fresh \subst{y}{u}{s}$ follows from
$z \fresh (s, u, y)$ by induction on $s$.
\qedhere
\end{itemize}
\end{proof}

\begin{definition}
We define $\beta$-reduction inductively by the \emph{compatible} 
closure~\cite[Definition~3.1.4]{Bare:84} of the $\beta$-rule (in its nominal version):
\Snippet[105]{beta_rule}
\end{definition}
\noindent
The freshness constraint on the $\beta$-rule is needed to obtain an induction 
principle strong enough with respect to avoiding capture of bound variables.
%
The following standard ``congruence properties''
(cf.~\cite[Lemma~3.1.6 and Proposition~3.1.16]{Bare:84})
will be used freely in the remainder:
\Snippet[153]{betas_cong}
They are proven along the lines of their textbook proofs, the first two by induction on the length
and the last one by (nominal) induction on $t$ followed by a nested (nominal) induction
on the definition of $\beta$-steps, using the Substitution Lemma.
Furthermore we will make use of the easily proven fact that $\beta$-reduction 
is \emph{coherent} with abstraction:
  \Snippet[217]{Abs_BetasE} 

\section{Z}

We present the Z-property for abstract rewriting, show that it implies confluence,
and then instantiate it for the case of (nominal) $\lambda$-terms modulo $\alpha$ 
equipped with $\beta$-reduction.

\begin{definition}
  A relation $\to$ on $A$ has the \emph{Z}-property if there is a
  map $\bullet : A \to A$ such that
  $a \to b \Longrightarrow b \to^* a^\bullet \wedge a^\bullet \to^* b^\bullet$.
\end{definition}

If $\to$ has the Z-property then it indeed is monotonic, i.e.,
$a \to^* b$ implies $a^\bullet \to^* b^\bullet$, which is straightforward to
show by induction on the length of the former.

\begin{lemma}
  A relation that has the Z-property is confluent.
\end{lemma}
\begin{proof}
  We show semi-confluence~\cite{Baad:Nipk:98}. So assume $a \to^* c$ and $a \to d$. We show
  $d \downarrow c$ by case analysis on the reduction from $a$ to $c$. If it is
  empty there is nothing to show. Otherwise there is a $b$ with $a \to^* b$ and
  $b \to c$. Then by monotonicity we have $a^\bullet \to^* b^\bullet$. From
  $a \to d$ we have $d \to^* a^\bullet$ using the Z-property, so in total
  $d \to^* b^\bullet$.  Since by applying the Z-property to $b \to c$ we also
  get $c \to^* b^\bullet$ we have $d \downarrow c$ as desired.
\end{proof}

There are two natural choices for functions on $\lambda$-terms that yield the
Z-property for $\to_\beta$, namely the full-development function
and the full-superdevelopment function.
The former maps a term  to the result of contracting all residuals of 
redexes in it~\cite[Definition~13.2.7]{Bare:84} and the latter 
also contracts the upward-created redexes, cf.~\cite[Section~2.7]{Raam:96}.
While Dehornoy and van Oostrom developed both proofs~\cite{Deho:Oost:08}, here
we opt for the latter, which requires less case analysis.

\begin{definition}
  We first define a variant of $\textit{App}$ with built-in $\beta$-reduction
  at the root:
  \Snippet[130]{app_beta}
\end{definition}
An easy case analysis on the first argument shows that this function satisfies
the congruence-like property \mbox{\snippet[279]{app_beta_Betas}}.

\begin{definition}
  \label{lambda_bullet}
  The full-superdevelopment function $\bullet$ on
  $\lambda$-terms is defined as follows:

  \Snippet[138]{lambda_bullet}
\end{definition}

Below, we freely use the fact that \snippet[238]{bullet_App}, which is shown by
considering whether or not $s^\bullet$ is an abstraction.
The structure of the proof that the $\lambda\beta$-calculus has the Z-property 
follows that for self-distributivity in that it build on the Self- and Rhs-properties.
The former expresses that each term \emph{self}-expands to its full-superdevelopment,
and the latter that applying $\bullet$ to the
\emph{right-hand side} of the $\beta$-rule, i.e., to the result of a substitution,
``does more'' than applying the map to its components first. Each is proven by structural induction.
\begin{lemma}[Self]
  \label{self}
  For all terms $t$ we have \snippet[167]{self}.
\end{lemma}
\begin{proof}
  By induction on $t$ using an additional case analysis on $t_1^\bullet$
  in the case that $t = t_1 \app t_2$.
\end{proof}
\begin{lemma}[Rhs] \label{lem:rhs}
  \label{rhs}
  For all terms $t$, $s$ and all variables $x$ we have \snippet[243]{rhs}.
\end{lemma}
\begin{proof}
  By induction on $t$. The cases $t = x$ and $t = \abs{y}{t'}$ are
  straightforward.  If $t = t_1 \app t_2$ we continue by case analysis on
  $t_1^\bullet$.

  If $t_1^\bullet = \abs{y}{u}$ then
  $\bsteps{\abs{y}{\subst{x}{s^\bullet}{u}} =
    \subst{x}{s^\bullet}{t_1^\bullet}}{\subst{x}{s}{t_1}^\bullet}$ by induction
  hypothesis. Then, using coherence of $\beta$-reduction with abstraction,
  we can obtain a term $v$ with
  $\subst{x}{s}{t_1}^\bullet = \abs{y}{v}$ and
  $\bsteps{\subst{x}{s^\bullet}{u}}{v}$.  We then have
  $\subst{x}{s^\bullet}{(t_1 \app t_2)^\bullet} =
  \subst{x}{s^\bullet}{\subst{y}{t_2^\bullet}{u}} =
  \subst{y}{\subst{x}{s^\bullet}{t_2^\bullet}}{\subst{x}{s^\bullet}{u}}$, using
  the substitution lemma in the last step. Together with
  $\bsteps{\subst{x}{s^\bullet}{u}}{v}$ and the induction hypothesis
  for $t_2$ this yields
  $\bsteps{\subst{x}{s^\bullet}{(t_1 \app t_2)^\bullet}}
    {\subst{y}{\subst{x}{s}{t_2}^\bullet}{v}}$. Since we also
  have $\subst{x}{s}{(t_1 \app t_2)}^\bullet =
  (\subst{x}{s}{t_1} \app \subst{x}{s}{t_2})^\bullet =
  \subst{y}{\subst{x}{s}{t_2}^\bullet}{v}$ we can conclude this case.

  If $t_1^\bullet$ is not an abstraction. then from the induction hypothesis we have
  $\subst{x}{s^\bullet}{(t_1 \app t_2)^\bullet} =
  \bsteps{\subst{x}{s^\bullet}{t_1^\bullet} \app \subst{x}{s^\bullet}{t_2^\bullet}}
    {\subst{x}{s}{t_1}^\bullet \app \subst{x}{s}{t_2}^\bullet \to_\beta^=
      \subst{x}{s}{(t_1 \app t_2)}^\bullet}$. 
   \end{proof}
\begin{lemma}[Z]
  The full-superdevelopment function $\bullet$ yields the Z-property
  for $\to_\beta$, i.e., we have \snippet[289]{lambda_Z} for all terms $s$ and $t$.
\end{lemma}
\begin{proof}
  Assume $\bstep{s}{t}$. We continue by induction on the derivation of
  $\to_\beta$.

  If $\bstep{s}{t}$ is a root step then $s = (\abs{x}{s'}) \app {t'}$ and
  $t = \subst{x}{t'}{s'}$ for some $s'$ and $t'$. Then
  $s^\bullet = \subst{x}{t'^\bullet}{s'^\bullet}$ and thus
  $\bsteps{t}{s^\bullet}$ using Lemma~\ref{self} twice, so
  $\bsteps{s^\bullet}{t^\bullet}$ by Lemma~\ref{rhs}.

  The case where the step happens below an abstraction follows
  from the induction hypothesis.

  If the step happens in the left argument of an application then
  $s = s' \app u$ and $t = t' \app u$. From the induction hypothesis and
  Lemma~\ref{self} we
  have $\bsteps{t' \app u}{\bstepe{s'^\bullet \app u^\bullet}{(s \app u)^\bullet}}$.
  That also $\bsteps{(s' \app u)^\bullet}{(t' \app u)^\bullet}$ follows directly
  from the induction hypothesis.
  The case where the step happens in the right argument of an application is
  symmetric.
\end{proof}

\section{Perspective}

This note originated from the bold and vague claim of
Dehornoy and van~Oostrom~\cite{Deho:Oost:08} that the confluence proof for the
$\lambda\beta$-calculus by establishing the Z-property for the
full-superdevelopment map, is \emph{the shortest}. We present a brief
qualitative and quantitative analysis of this claim.

Three major methods in the literature for showing confluence of the $\lambda\beta$-calculus are:
\[ \mbox{complete developments} \models \diamond 
   \,\Longrightarrow\,
   \mbox{complete, full-developments} \models \angle
   \,\Longleftarrow\,
   \mbox{full-developments} \models \mbox{Z} \]
From left to right, that complete developments have the diamond ($\diamond$) property is 
due to Tait and Martin--L\"of~\cite[Section~3.2]{Bare:84},
that complete developments have the angle ($\angle$) property with respect to the full-development
function is due to Takahashi~\cite{T95} (cf.~\cite[Proposition~1.1.11]{Tere:03}),
and that full-developments have the Z-property is due to~\cite{Deho:Oost:08}.
From the fact that the second method needs the concepts of both the others,
it stands to reason that its formalization is not the shortest, 
as confirmed by a formalization of Nipkow~\cite{N01} and our quantitative analysis below.

Our proof varies on the above picture along yet another dimension, replacing
developments (due to Church and Rosser, cf.~\cite[Definition~11.2.11]{Bare:84})
by superdevelopments (due to Aczel, cf.~\cite[Section~2.7]{Raam:96}).
Where full-developments give a ``tight'' upper bound on the single-step
reducts of a given term, full-superdevelopments do not, and one may hope for
a simplification of the analysis because of it. This is confirmed
by our quantitative analysis below.
One may vary along this dimension as well: \emph{any} map $\bullet$ 
having the Z-property suffices as we show now.
\begin{definition}
  A relation $\to$ on $A$ has the \emph{angle} property for
  a map $\bullet$ from $A$ to $A$, and relation $\dstep$ on $A$,
  if ${\to} \subseteq {\dstep} \subseteq {\to^*}$ and
  $a \dstep b$ implies $b \dstep a^\bullet$.
\end{definition}
\begin{lemma}
  A relation $\to$ has the Z-property for map $\bullet$ if and only if it
  has the angle property for map $\bullet$ and some relation $\dstep$.
\end{lemma}
\newpage
\begin{proof}
  First assume that $\dstep$ has the angle property for map $\bullet$ and relation $\dstep$.
  To show that $\to$ has Z assume $a \to b$. Then by assumption we also have $a \dstep b$ and
  hence $b \dstep a^\bullet$ and $a^\bullet \dstep b^\bullet$, by applying the
  angle property twice, which together with ${\dstep} \subseteq {\to^*}$
  yields Z.

  Now assume $\to$ has the Z-property. We define the \emph{$\bullet$-development}\footnote{%
For the full-development map $\bullet$ such \emph{syntax-free} $\bullet$-developments may differ
from the usual ones, e.g.\ for $(\abs{y}{I}) \app ((\abs{x}{x \app x}) \app I)$.
We conjecture that on terminating, non-erasing and non-collapsing $\lambda$-terms they coincide.
}
  relation by $a \dstep b$ if $a \to^* b$ and $b \to^* a^\bullet$. 
  Then ${\to} \subseteq {\dstep} \subseteq {\to^*}$ follows
  from the definition of $\dstep$ and the Z-property.  The angle
  itself directly follows from the definition of $\dstep$ and monotonicity of
  $\bullet$.
\end{proof}
We turn to the quantitative analysis of the claim of~\cite{Deho:Oost:08}.
Formalizing confluence of the $\lambda\beta$-calculus has a long history
for which we refer the reader to~\cite{N01}.
We compare our formalization in Isabelle to two other such,
Nipkow's formalization in Isabelle/HOL~\cite{N01} (as currently distributed with Isabelle)
and Urban and Arnaud's formalization in Nominal Isabelle.\footnote{%
\scriptsize
\url{http://www.inf.kcl.ac.uk/staff/urbanc/cgi-bin/repos.cgi/nominal2/file/d79e936e30ea/Nominal/Ex/CR.thy}}
There are two major differences of the present proof to Nipkow's formalization.
On the one hand Nipkow uses de~Brujin indices to represent $\lambda$-terms. 
This considerably increases the size of the formal theories -- almost 200 lines of 
the roughly 550 line development are devoted to setting up terms and the required manipulations
on indices. Our development is 300 lines (60 of which are used for our ad hoc
Eisbach methods). The second difference is the actual technique used to show
confluence: Nipkow proceeds by establishing the diamond property for 
complete developments.
Urban and Arnaud proceed by establishing the
angle property for multisteps with respect to 
the full-development function.
This results in a 100 line increase compared to our formalization.

\newcommand{\doi}[1]{\href{http://dx.doi.org/#1}{doi:\nolinkurl{#1}}}\newcommand{\afp}[1]{\url{https://www.isa-afp.org/entries/#1.shtml}}\providecommand{\noopsort}[1]{}

\end{document}